\documentclass[12pt,reqno]{amsart}
\usepackage{graphicx}
\usepackage{hyperref}
\usepackage{amsmath,amssymb}
\usepackage[bbgreekl]{mathbbol}

\large\normalsize

\theoremstyle{plain}
\newtheorem{theorem}{Theorem}[section]
\newtheorem{lemma}[theorem]{Lemma}

\theoremstyle{definition}

\numberwithin{equation}{section}
\numberwithin{equation}{section}

\newcommand{\C}{\mathbb{C}}
\newcommand{\R}{\mathbb{R}}
\newcommand{\N}{\mathbb{N}}

\newtheorem{cor}[theorem]{Corollary}

\newtheorem{prop}[theorem]{Proposition}

\newcommand{\ip}[2]
{\ensuremath{\langle #1,#2 \rangle}}

\begin{document}
\title[Discrete Scr{\"o}dinger Operators]{Half Line Titchmarsh-Weyl $m$ Functions of Vector-Valued Discrete Schr{\"o}dinger Operators}
	
\author[Acharya]{Keshav Raj Acharya}
\address{Department of Mathematics,
Embry-Riddle Aeronautical University,
1 Aerospace Blvd., Daytona Beach, FL 32114, U.S.A.}
\email{acharyak@erau.edu}

\author[McBride]{Matt McBride}
\address{Department of Mathematics and Statistics,
Mississippi State University,
175 President's Cir., Mississippi State, MS 39762, U.S.A.}
\email{mmcbride@math.msstate.edu}

\thanks{}
\date{\today}

\begin{abstract}
	 We show that the half-line $m$ functions associated with the vector-valued Schr\"odinger operators are the elements in the Siegel upper half space. We introduce a metric on the space of $m$ functions associated to the vector-valued discrete Schr\'odinger operators. Then we show that the action of transfer matrices on these $m$ functions is distance decreasing.  
\end{abstract}
\maketitle
\thispagestyle{empty}
	
\section{Introduction}
The spectral theory of one dimensional Schr\"odinger operator has been well studied. The extension of the theory to higher dimensions has yet to be developed and is of interest of many researchers.  There is some literature which deals with the Schr\"odinger equations in higher dimensions with matrix-valued potentials, for example see \cite{SF, CGH, JFRA, DD, JG}. Likewise, as the theory of Titchmarsh-Weyl $m$ functions provide elegant methods in describing the spectrum of the associated operators, it is natural to study these functions for vector-valued discrete Schr\"odinger operators. Some basic theory has been developed in \cite{KA1}. In order to study these $m$ functions we need to study matrix-valued Herglotz functions. In this regard, there has been some studies about the matrix-valued Herglotz functions, for example see \cite{FG}. Our goal in this paper is to discuss some basic properties of half-line Titchmarsh-Weyl $m$ functions associated to the vector-valued discrete Schr\"odinger operators. 

We begin by considering an equation of the form,
\begin{equation}\label{ds}
y(n+1)+y(n-1)+B(n)y(n)=zy(n)\,,\, z\in \C
\end{equation}
where for fixed $n$, $y(n) =(y_1(n),y_2(n),\ldots\,,y_d(n))^t \in \C^d$ and $B(n)\in\R^{d\times d}$ is a symmetric matrix. Here $\{y(n)\}$ and $\{ B(n)\}$ are respectively vector-valued and matrix-valued sequences. One can consider $B(n)$ as a Hermitian matrix but for the sake of simplicity we only consider $B(n)$ as a real symmetric matrix.  In this paper, we give insight into half-line $m$ functions and discuss  properties these functions enjoy. Then we describe the action of transfer matrices associated to equation \eqref{ds} on these functions. We then introduce the metric on the space of $m$ functions an extended analogous of hyperbolic metric in complex upper half plane. Finally we observe the distance properties. The paper is organized as follows: We present the preliminaries in Section 2. In Section 3, we define the half-line $m$ functions following the definition in \cite{KA1}, present very important properties. Then we define a metric on the space of $m$ functions and show that action of transfer matrices on $m$ functions is distance decreasing.   

\section{Preliminaries}
Similar to the one dimensional space, we consider the equation \eqref{ds} in a Hilbert space. 	 
Let $\ell^2(\N,\C^d)$ be the Hilbert space of square summable vector-valued sequences. That is 
\begin{equation*}
\ell^2(\N,\C^d) = \left\{y: \sum_{n=1}^{\infty} y(n)^*y(n) < \infty\right\}\,.
\end{equation*} 

Equation (\ref{ds}) induces a vector-valued discrete Schr\"odinger operator $J$ on $\ell^2(\N,\C^d)$ as 
\begin{equation*}
(Jy)(n)=\left\{
\begin{aligned}
&y(n+1)+y(n-1)+B(n)y(n) &&\textrm{if }  n>1 \\
&y(1) + B(1) y(2) &&\textrm{if } n=1\,.
\end{aligned}\right.
\end{equation*}	 
Notice from this we see that $J$ is a bounded self-adjoint operator.	

Define the difference expression $\tau: \ell^2(\N,\C^d) \to \ell^2(\N,\C^d)$ by
\begin{equation}
\tau(u(n)) = u(n+1)+u(n-1) + B(n)u(n)
\end{equation}
which looks like $J$ but acts on any vector-valued sequence. In fact, for any $z\in\C$, a solution $u(n)$ of $ (\tau-z)u = 0 $ is uniquely determined by the values $u(n_0)$ and $u(n_0 +1)$, therefore the solution space is a $2d$ dimensional vector space of sequences.	It is shown in \cite{KA1} that there are precisely $d$ linearly independent solutions to equation (\ref{ds}) that are in $\ell^2 (\N, \C^d)$.
	
One first observation is to connect equation (\ref{ds}) with a second order differential expression, similar to the one in one dimensional space, see \cite{GT}.   Indeed, define the following first order difference operator $\Delta:\ell^2(\N,\C^d)\to \ell^2(\N,\C^d)$ by
\begin{equation*}
(\Delta u)(n) = u(n+1) - u(n).
\end{equation*}
An immediate calculation yields the adjoint to $\Delta$, namely $\Delta^*:\ell^2(\N,\C^d)\to \ell^2(\N,\C^d)$ via
\begin{equation*}
(\Delta^* u)(n)= u(n-1)- u(n)\,.
\end{equation*}	

Using the definition of $\tau$, consider the following calculations:
\begin{equation*}
\begin{aligned}
-(\Delta^*\Delta u)(n) &= -\Delta^*(u(n+1) - u(n)) = -u(n) +u(n+1) +u(n-1) - u(n) \\ 
& = (\tau u)(n)  - (B(n) + 2I)u(n)\,,
\end{aligned}
\end{equation*}
and
\begin{equation*}
\begin{aligned}
-(\Delta\Delta^*u)(n) & = -\Delta( u(n-1) -u(n)) = u(n+1) + u(n-1) - 2u(n) \\ 
&=(\tau u)(n)-(B(n) + 2I)u(n)\,.
\end{aligned}
\end{equation*}	
Notice that this is similar to the vector-valued Sturm-Liouville differential expression
\begin{equation*}
-\frac{d^2}{dx^2}(y) + V(x)y
\end{equation*}
where $y$ is in some properly chosen domain and $V$ is a potential function.
	
For any two vector valued sequences $\{x(n)\}, \{v(n)\}$, recall that the \textit{Wronskian} is defined by
\begin{equation} 
W_n(x(n),y(n))= x(n+1)^t y(n)- x(n)^t y(n+1)\,.
\end{equation}
It can be easily shown that if $u(n)$ and $v(n)$ are any two solutions to equation (\ref{ds}), then $W_n(u(n),v(n))$ is independent of $n$. This definition can be easily generalized to $d\times d$ matrix-valued sequences $\{X(n)\}$ and $\{Y(n)\}$ as follows:
\begin{equation} 
W_n(X(n),Y(n))= X(n+1)^t Y(n)- X(n)^t Y(n+1)\,.
\end{equation}
In this case $W_n(X(n),Y(n))$ is a $d\times d$ matrix-valued function.  It was shown in \cite{KA1} that if $X(n)$ and $Y(n)$ are matrix valued solutions to equation (\ref{ds}), then is $W_n(X(n),Y(n))$ independent of $n$. Using Abel summation we arrive at Green's formula:
\begin{equation}\label{gi} 
\sum_{j=n}^m \left(X(j)^*(\tau Y)(j) - (\tau Y)(j)^*Y(j)\right) = W_{n-1}(\overline{X}(n-1),Y(n-1))- W_m(\overline{X}(m),Y(m)) \,.
\end{equation}
Please see \cite{KA1} for a detailed computation of formula \eqref{gi}.
	
\section{Half line Weyl $m$ functions}
A $d\times 2d$ matrix-valued sequence $Y(n)$ is said to be a fundamental matrix solution to equation (\ref{ds}) if it satisfies equation (\ref{ds}) in matrix form and columns of $Y(n)$ form a set of linearly independent solutions to the vector form to equation (\ref{ds}). 

Write $Y = (U, V)$ where $U$ and $V$ are $d\times d$ matrix-valued solutions to equation (\ref{ds}).  One of the goals of this paper is to discuss the half line $m$ function.  Thus, we consider the equation given by (\ref{ds}) on $\N_0 = \N \cup\{0\}$. Let $\N_- = \{0, 1,\ldots n\}$ and $\N_+ =\{ n+1, N+2,\ldots\}$. Suppose $U$ and $V$ satisfy the following initial conditions at $n$: 	 
\begin{equation*}
\begin{aligned}
&U(n)  =  - I, && V(n) = 0 \\ 
&U(n+1)= 0,  && V(n+1) = I\,.
\end{aligned}
\end{equation*}
Define the Weyl $m$ functions on $\N_-$ and $\N_+$ by	
\begin{equation}\label{ss}
F_{\pm}(n) = U(n) \pm  M_{\pm}(n) V(n)
\end{equation}
where we require that $F_-(0) = 0$ and $F_+\in\ell^2(\N,\C^{d \times d})$. Notice that the columns of $F_{\pm}(n)$ form a linearly independent set of solutions to equation (\ref{ds}) and therefore for each $n$, $F_{\pm}(n)$ are invertible.  For $z\in\C^+$, these half-line $m$ functions are uniquely determined and are given by:
\begin{equation} \label{mf}
M_+(n,z) = - F_+(n+1,z)F_+(n,z)^{-1},\ \  M_-(n,z)  =  F_-(n+1,z)F_-(n,z)^{-1}\,.
\end{equation}
The theory of these $m$ functions has been very useful tool in the Spectral theory of Schr\"odinger operators. There are numerous literatures dealing with these $m$ functions for one dimensional  Schr\"odinger operators, see \cite{GT} as an example, for Jacobi operators. There are also some articles which discuss about the Weyl $m$ functions for matrix-valued Schr\"odinger operators in continuous case, see \cite{SF, CGH,DD, JG,FG}. Our goal is to discuss some basic properties of the $m$ functions defined by \eqref{mf}.

First we show that $M_{\pm}(n)$ can be expressed in terms of a resolvent operator. As such we introduce the following notation.  Let $J_{\pm}$ be the restriction of $J$ to the space $\ell^2(\N_{\pm},\C^d)$ and denote the Dirac delta type vector-valued sequences as follows:
\begin{equation*}
\begin{aligned}
&\bbdelta_1^1 = \left \{\begin{bmatrix}1\\0\\  \vdots \\ 0 \end{bmatrix}, \begin{bmatrix}0\\0\\  \vdots \\ 0 \end{bmatrix}, \dots  \right\}, \bbdelta_2^1 = \left \{\begin{bmatrix}0\\1\\  \vdots \\ 0 \end{bmatrix}, \begin{bmatrix}0\\0\\  \vdots \\ 0 \end{bmatrix},  \dots  \right\},\ldots,\\
&\bbdelta_d^1 = \left \{ \begin{bmatrix}0\\0\\  \vdots \\ 1 \end{bmatrix}, \begin{bmatrix}0\\0\\  \vdots \\ 0 \end{bmatrix}, \dots  \right\},\ldots,\bbdelta_d^n = \left \{ \begin{bmatrix}0\\0\\  \vdots \\ 0 \end{bmatrix},  \dots   \begin{bmatrix}0\\0\\  \vdots \\ 1 \end{bmatrix}, \dots \right\}
\end{aligned}
\end{equation*}

Similarly for matrix-valued sequence denote:
\begin{equation*}
\mathbb{\Delta}_1 = \left \{I,0\ldots\right\},\ldots,\mathbb{\Delta}_n = \left \{0,\ldots,I,\ldots\right\}
\end{equation*}
where $I$ and $0$ are the $d\times d$ identity and null matrix respectively. We have the following proposition:
\begin{prop}
Let $z\in \C^+$.
\begin{equation}\label{mr}
M_+(n,z) = \ip{\mathbb{\Delta}_{n+1}}{ (J_{+}-zI)^{-1} \mathbb{\Delta}_{n+1}} = \Big( \ip{ \bbdelta_i^{n+1}}{ (J_{+}-zI)^{-1} \bbdelta_j^{n+1}} \Big)
\end{equation}
\end{prop}

\begin{proof}
Let $g_j(k,z) =  (J_ +-zI)^{-1} \bbdelta_j^{n+1}(k)$.   Notice that $g_j(k,z)$, for $j = 1,\ldots, d$ and for $k>n+1$ are solutions to equation (\ref{ds}) and belong to $\ell^2(N_+,\C^d )$.  Define $G(k,z)$ by
\begin{equation*}
G(k,z)= \Big(g_1(k,z), g_2(k,z), \dots, g_d(k,z)\Big)\,.
\end{equation*}
From this definition it follows that $G(k,z)$ is a matrix valued solution to equation (\ref{ds}) for $k>n+1$ and $G(k,z)\in\ell^2(N_+,\C^{d\times d})$. Thus there exists an invertible constant matrix $C\in\C^{d\times d}$ such that 
\begin{equation}\label{c1}
G(k,z) = F(k,z)C\,.
\end{equation}
Since $F(n,z) = -I$ and $\ G(n,z) = - C$, it follows that 
\begin{equation*}
M_+(n,z) = - F_+(n+1,z)F(n,z)^{-1} =  -G(n+1,z) G(n,z)^{-1} =  G(n+1,z) C^{-1}\,.
\end{equation*}
To find the matrix $C$, we need only to compare the values at $n+2$ place.  Notice that 
\begin{equation*}
(J_+ -z I)G(n+1, z) = (\bbdelta_1^1, \bbdelta_2^1, \dots \bbdelta_d^1)
\end{equation*}
and hence $(J_+ -z I)G(n+1, z) = I$. Therefore,
\begin{equation}\label{c2}
G(n+2,z) = (z-B(n+1))G(n+1, z) +I\,.
\end{equation}
Moreover, since $F(n,z)$ is a solution to equation (\ref{ds}) we have
\begin{equation*}
F(n+2, z) = (z-B(n+1))F(n+1, z) - F(n, z)\,.
\end{equation*}
From equation (\ref{c1}) it follows that
\begin{equation}\label{c3}
G(n+2, z) = (z-B(n+1))F(n+1, z)C -F(n,z)C\,.
\end{equation}
Finally, comparing equations (\ref{c2}) and (\ref{c3}) we deduce that $-F(n,z)C = I$. Since $F(n,z) = -I$, it follows that $C = I$.
 Therefore, we have that $M_+(n,z) = G(n+1,z)$. 

Thus the matrix  $M_+(n,z)$ by comparing the values at $k = n+1$, $$ G(n+1,z) = F(n+1, z) C$$  

Now we find the entries of the matrix solution $G(n+1,z)$. Indeed, for each $j= 1,\ldots,d$, the components $g_{ij(n+1,z)}$ for 
\begin{equation*}
g_j(n+1, z) = \begin{bmatrix}g_{1j}(n+1,z)\\g_{2j}(n+1,z)\\  \vdots \\ g_{dj} (n+1,z) \end{bmatrix}
\end{equation*}
are given by 
\begin{equation*}
g_{ij(n+1,z)}= \ip{ \bbdelta_i^{n+1}}{ (J_{+}-zI)^{-1} \bbdelta_j^{n+1}}
\end{equation*}
for $i =1,\ldots,d$.  It follows that 
\begin{equation*}
M_+(n,z) = \ip{\mathbb{\Delta}_{n+1}}{ (J_{+}-zI)^{-1} \mathbb{\Delta}_{n+1}} = \Big( \ip{ \bbdelta_i^{n+1}}{ (J_{+}-zI)^{-1} \bbdelta_j^{n+1}} \Big)\,.
\end{equation*}
This completes the proof.
\end{proof}
 Eqn. \eqref{mr} suggests that $ \overline {M_+(n,z)} = M_+(n,\bar{z}).$ \\ 
We now show that $M_{\pm}(n,z)$ are matrix-valued Herglotz functions. 
 
\begin{lemma}\label{lem2.2}
For any $z \in \C^+$ and any $n$, $M_{\pm}(n,z)$ are symmetric and satisfy the relation
\begin{equation}\label{rr}
M_{\pm}(n,z) + M_{\pm}(n \mp1,z)^{-1} \mp \left(zI-B(n) \right) = 0\,. 
\end{equation}
\end{lemma}
 
\begin{proof}
We first apply the Wronskian to show that $M_{\pm}(n)$ are symmetric.  Indeed, for $M_-(n,z)$, consider the following calculation:
\begin{equation*}
\begin{aligned}
&M_-(n,z)^t - M_-(n,z) = \Big(F_-(n+1,z)F_-(n,z)^{-1} \Big)^t - F_-(n+1,z)F_-(n,z)^{-1} \\ 
& = \Big(F_-(n,z)^{-1} \Big)^t F_-(n+1,z)^t - F_-(n+1,z)F_-(n,z)^{-1} \\ 
& = \Big(F_-(n,z)^{-1} \Big)^t \Big(F_-(n+1,z)^t F_-(n,z) - F_-(n,z)^t F_-(n+1,z)\Big) F_-(n,z)^{-1}  \\  
&= \Big(F_-(n,z)^{-1} \Big)^t W_n (F_-(n,z), F_-(n,z)) F_-(n,z)^{-1}\,. 
\end{aligned}
\end{equation*}  
Since $F_-(n,z)$ is a solution to equation (\ref{ds}), $W_n(F_-(n,z),F_-(n,z))$	is independent of $n$. It therefore follows that $W_n (F_-(n,z),F_-(n,z)) = W_0(F_-(0,z),F_-(0,z))= 0$.  This in turn implies that $M_-(n,z)$ is symmetric. A similar calculation shows that $M_+(n,z)$ is also symmetric. All that remains is to show that $M_\pm(n,z)$ satisfies equation (\ref{rr}).

Indeed, from equation (\ref{ds}), we have:
\begin{equation*}
F_+(n+1,z) + F_+(n-1,z) +\big(B(n) -zI\big)F_+(n,z) =  0
\end{equation*} 
and hence 
\begin{equation*}
F_+(n+1,z) F_+ (n,z)^{-1}  + F_+(n-1,z) F_+ (n,z)^{-1} +\big(B(n) -z I\big) = 0\,.
\end{equation*}
This implies that 
\begin{equation*}
M_+(n,z) + M_+(n-1,z)^{-1} + zI-B(n) = 0\,.
\end{equation*}
An analogous calculation shows that for $M_-(n,z)$ we have
\begin{equation*}
M_-(n,z) + M_-(n-1,z)^{-1} + B(n)-zI = 0\,.
\end{equation*}
This completes the proof.
\end{proof}

Recall the imaginary part of $ M_{\pm}(n,z)$ is given by:
\begin{equation*}
\operatorname{Im} M_{\pm}(n,z) =  \frac{1}{2i}\left(M_{\pm}(n,z) -  M_{\pm}(n,z)^*\right)\,.
\end{equation*}
We have the following proposition:
\begin{prop}\label{prop2.3}
For any $z\in \C^+$, the matrices $\operatorname{Im} M_{\pm}(n,z)$ are positive.
\end{prop}
 
\begin{proof}
First we show that $\operatorname{Im} M_-(n,z)$ is a positive matrix. Using Green's formula, equation (\ref{gi}) with $X(n)$ and $Y(n)$ are replaced by $F_-(n,z)$ we get 
\begin{equation}\label{gi1}
\begin{aligned}
&\sum_{j=1}^n\left(F_-(j,z)^*(\tau F_-)(j,z)-(\tau F_-)(j,z)^* F_-(j,z) \right) = \\
&\qquad=W_{0}(\overline{F_-(0,z)},F_-(0,z))- W_n(\overline{F_-(n,z)},F_-(n,z))\,.
\end{aligned}
\end{equation}
Since $F_-(0,z) = 0$ and $(\tau F_-)(j,z) = zF_-(j,z)$, equation (\ref{gi1}) simplifies to 
\begin{equation}\label{gi2} 
(z - \overline{z})\sum_{j=1}^n F_-(j,z)^*F_-(j,z) = -W_n(\overline{F_-(n,z)},F_-(n,z))\,.
\end{equation}
Since $F_-(n,z)$ can be decomposed into the fundamental solutions $U(n,z)$ and $V(n,z)$, that is $F_-(n, z) = U(n,z) - M_-(n,z) V(n,z)$, the linearity of Wronskian implies that
\begin{equation}\label{gi3}
\begin{aligned}
&(z -\bar{z})\sum_{j=1}^n  F_-(j,z)^*F_-(j,z) = \\
&=- W_n(\overline{U(n,z) - M_-(n,z) V(n,z)},U(n,z) - M_-(n,z) V(n,z))  \\
& =  W_n(U(n, \overline{z}), U(n,z) ) + W_n(U(n, \overline{z}),  M_-(n,z) V(n,z) ) \\
&\qquad + W_n( M_-(n,\overline{z})V(n,\overline{z}),U(n,z)) - W_n(M_-(n,\overline{z})V(n,\overline{z}), M_-(n,z) V(n,z))\,.
\end{aligned}
\end{equation}
Using definition of the Wronskian and the boundary conditions at $n$ and $n+1$ we obtain that 

\begin{equation*}
\begin{aligned}
&W_n(U(n,\overline{z}),U(n,z)) = 0\\
&W_n(M_-(n,\overline{z})V(n, \overline{z}),M_-(n,z) V(n,z)) = 0\,.
\end{aligned}
\end{equation*}
It follows that 
\begin{equation*}
\begin{aligned}
&W_n(U(n,\overline{z}), M_-(n,z) V(n,z) ) = M_-(n,z)\\ 
&W_n(M_-(n, \overline{z})V(n,\overline{z}),U(n,z)) = -M_-(n, \overline{z})\,.
\end{aligned}
\end{equation*}
Hence equation (\ref{gi3}) becomes
\begin{equation*}
(z - \overline{z})\sum_{j=1}^n F_-(j,z)^*F_-(j,z) = M_-(n,z) - M_-(n, \overline{z})\,. 
\end{equation*}
Since $z\in\C^+$, we have $(z - \overline{z})/2i >0$ and thus it follows that $\operatorname{Im} M_-(n,z)$ is a positive matrix.

Similarly, to show $\operatorname{Im} M_+(n,z)$ is a positive matrix, we replace $X(n)$ and $Y(n)$ both with $F_+(n,z)$ in equation (\ref{gi}). Since $F_+(n,z)$ is defined on $\N_+$, for $N\ge n+1$ we have
\begin{equation}\label{gi4}
\begin{aligned}
&\sum_{j=n+1}^N \left(F_+(j,z)^*(\tau F_+)(j,z) - (\tau F_+)(j,z)^*F_+(j,z)\right)\\
&\quad= W_{n}(\overline{F_+(n,z)},F_+(n,z))- W_N(\overline{F_+(N,z)},F_+(N,z))\,.
\end{aligned}
\end{equation}

Since $F_+(z)\in \ell^2(\N,\C^{d\times d})$ for every $z\in\C^+$, taking the limit as $N \rightarrow \infty$ yields zero on the left-hand side of equation (\ref{gi4}).  Moreover since the Wronskian of solutions to equation (\ref{ds}) is independent of $N$ we get 
\begin{equation*}
\lim_{N\to\infty}W_N(\overline{F_+(N,z)},F_+(N,z)) = 0\,.
\end{equation*} 
Thus using a similar argument as above it follows that
\begin{equation*}
(z - \overline{z})\sum_{j=1}^n F_+(j,z)^*F_+ (j,z) = M_+(n,z) - M_+(n, \overline{z})\,.
\end{equation*}
Therefore $\operatorname{Im} M_+(n,z)$ is a positive matrix completing the proof.
\end{proof}
The fractional linear transformation have been used in analyzing the theory of $m$ functions. More specifically, the transfer matrices associated to the Schr\"odinger equations are fractional linear transformation. For   the matrix-valued Schr\"odhinger operators, the transfer matrices are considered as matrix-valued fractional linear transformation.
For any $T\in\C^{2d \times 2d}$ of the form:
\begin{equation*}
T= \begin{pmatrix} A & B \\ C & D \end{pmatrix}
\end{equation*}
where are $A$, $B$, $C$, and $D$ are $d\times d$ matrices, a {\it matrix-valued fractional linear transformation} is a map $T:\C^{d\times d}\rightarrow\C^{d\times d}$ defined by
\begin{equation}\label{flt}
T(Z) = (AZ + B) (CZ +D)^{-1}\,,
\end{equation}
for all $Z\in \C^{d\times d}$ for which $T$ is well defined. We observe that the transfer matrices associated to the equation \eqref{ds} are in fact matrix-valued fractional linear transformation of the form \eqref{flt}.

If $F(n)$ is a matrix solution to equation (\ref{ds}) we have 
\begin{equation}\label{tm}
\begin{bmatrix} F(n+1) \\ \mp F(n) \end{bmatrix} = \begin{pmatrix} z I -B(n) & \pm I \\ \mp I  & 0 \end{pmatrix} \begin{bmatrix} F(n) \\ \mp F(n-1) \end{bmatrix}\,,
\end{equation}
where $B(n)$ is a $d\times d$ symmetric matrix.  The matrices 
\begin{equation}\label{tm1} 
T_{\pm}(n) = \begin{pmatrix} z I -B(n) & \pm I \\ \mp I  & 0 \end{pmatrix}
\end{equation}
given in equation (\ref{tm}) are called {\it transfer matrices} which describe the evolution of the vectors
\begin{equation*}
\begin{bmatrix} F(n+1) \\ \mp F(n) \end{bmatrix}
\end{equation*}
under iteration of $T_{\pm}(n)$.  These transfer matrices $T_{\pm}(n)$ can be considered as complex matrix-valued linear transformations \eqref{flt} acting on the space of $m$ functions by using \eqref{tm1}:
\begin{equation*}
T_{\pm}(n)M_{\pm}(n) = (z I -B(n))(M_{\pm}(n) \pm I)(\mp I M_{\pm}(n))^{-1}\,.
\end{equation*}
 The  transfer matrices also relate the $m$ functions at $n$ and $n-1$ as shown in the following lemma.
\begin{lemma}\label{mfunctrans}
For any $z\in \C^+$, $M_{\pm}(n) = T_{\pm}(n)M_{\pm}(n-1)$. 
\end{lemma}

\begin{proof}  
Since $F_{\pm}(n)$ are solutions to equation (\ref{ds}) we have,	
\begin{equation*}
F_{\pm}(n+1,z) + F_{\pm}(n-1,z) + B(n) F_{\pm}(n,z)= zF_{\pm}(n,z)\,.
\end{equation*}  
Consider the following calculation:
\begin{equation*}
\begin{aligned}
\begin{pmatrix}
z-B(n) & I \\ -I & 0
\end{pmatrix}M_+(n-1) &= \begin{pmatrix}
z-B(n) & I \\ -I & 0
\end{pmatrix} \begin{bmatrix} F_+(n,z) \\ -F_+(n-1,z)\end{bmatrix} \\
&=\begin{bmatrix}
(z-B(n))F_+(n,z) - F_+(n-1,z) \\ -F_+(n,z)
\end{bmatrix} = \begin{bmatrix} F_+(n+1,z) \\ -F_+(n,z) \end{bmatrix}\\
&=-F_+(n+1,z)F_+(n,z)^{-1} = M_+(n)\,
\end{aligned}
\end{equation*}
where the last line follows from the usual interpretation of matrix-valued linear transformations.  This shows the lemma is true for $M_+(n)$.  A similar calculation shows
\begin{equation*}
M_-(n) = T_-(n)M_-(n-1)\,.
\end{equation*} 
This completes the proof.
\end{proof}

The   matrices $T_{\pm}(n)$  satisfy the following symplectic identity:

\begin{lemma}\label{sm}
For all $n\in\N$ and any symmetric $d\times d$ matrix $B$,  $T_{\pm}(n)$ satisfy the following identities:
\begin{enumerate}
\item $ T_{\pm}(n)^t \begin{pmatrix} 0 & I \\ -I &0\end{pmatrix}T_{\pm}(n) = \begin{pmatrix} 0 &  I \\ -I  & 0 \end{pmatrix}$.
\item $\begin{pmatrix} I &  0 \\ 0  & -I \end{pmatrix}T_+(n) \begin{pmatrix} I &  0 \\ 0  & -I \end{pmatrix} = T_-(n)$.
\end{enumerate}
\end{lemma}  
 
\begin{proof} These identities follow from matrix multiplication.  Indeed we have
\begin{equation*}
\begin{aligned} T_{\pm}(n)^t \begin{pmatrix}0 & I \\ -I & 0\end{pmatrix}T_{\pm}(n)  & =  \begin{pmatrix} z I -B(n) & \pm I \\ \mp I  & 0 \end{pmatrix}^t\begin{pmatrix} 0 &  I \\ -I  & 0 \end{pmatrix}  \begin{pmatrix} z I -B(n) & \pm I \\ \mp I  & 0 \end{pmatrix} \\ 
& = \begin{pmatrix} z I -B(n) & \mp I \\ \pm I  & 0 \end{pmatrix} \begin{pmatrix}  \mp I  & 0 \\ -z I + B(n)  & \mp I \end{pmatrix}\\ 
& = \begin{pmatrix} 0 &  I \\ -I  & 0 \end{pmatrix}\,.
\end{aligned}
\end{equation*}
This shows the first identity.  Similarly we have
\begin{equation*}
\begin{aligned}
\begin{pmatrix} I &  0 \\ 0  & -I \end{pmatrix} T_+(n) \begin{pmatrix} I &  0 \\ 0  & -I \end{pmatrix} & = \begin{pmatrix} I &  0 \\ 0  & -I \end{pmatrix} \begin{pmatrix} z I -B(n) & I \\ - I  & 0 \end{pmatrix}\begin{pmatrix} I &  0 \\ 0  & -I \end{pmatrix} \\ 
& = \begin{pmatrix} z I -B(n) &  -I \\ I  & 0 \end{pmatrix}\,,
\end{aligned} 
\end{equation*} 
showing the second identity and completing the proof.
\end{proof}

The matrix in the first identity is a special symplectic matrix and often appears.  We give it the following label:
\begin{equation*}
J = \begin{pmatrix} 0 &  I \\ -I  & 0 \end{pmatrix}\,.
\end{equation*} 

It follows from Lemma \ref{sm} that $T_{\pm}(n)\in SL(2d,\C)$ which is the group of $2d \times 2d $ complex symplectic matrices. Therefore, these are matrix-valued functions mapping complex upper half plane $\C^+$ to $SL(2d,\C).$  From Lemma \ref{lem2.2} and Proposition \ref{prop2.3} that for any $n$ and any $z \in \C^+, $  $M_{\pm}(n, z) $ are symmetric and  have positive definite imaginary part. Thus these are matrices in $\mathcal{S}_d$, the Siegel upper half plane:
\begin{equation*}
\mathcal{S}_d = \{Z\in\C^{d \times d} : Z = X +iY,\, X^t=X,\, Y^t = Y,\, Y> 0 \}\,.
\end{equation*}
 
If $T\in SL(2d,\R)$, the mapping $Z\mapsto T(Z)$ in equation (\ref{flt}) is a well defined map and is a group action on $\mathcal{S}_d$, see \cite{KAM} for details. The transfer matrices $T_{\pm}(n)$, however, are complex symplectic matrices acting on the space of $m$ functions. \cite {KAM} provided a necessary condition for a complex symplectic matrix $T\in SL(2d,\C)$ so that the map $Z\mapsto T(Z)$ is well defined.  Here we show that the map $Z\mapsto T_{\pm}(n,z)(Z)$ on $\mathcal{S}_d$ is well defined.   More precisely we have the following lemma.
 
\begin{lemma}\label{lem2.6}
For any $z\in\C^+$, if $B(n)$ is any real symmetric matrix, then $T_{\pm}(n,z)(Z)\in\mathcal{S}_d$ for all $Z\in\mathcal{S}_d$.
\end{lemma}
 
\begin{proof} 
For any $z\in\C^+$, the second identity in Lemma \ref{sm}, implies $T_+(n,z)$ and $T_-(n,z)$ are conjugate to each other. Therefore, it is enough to show that $ T_-(n,z)(Z)\in\mathcal{S}_d$. Write the transfer matrix $T_-(n,z)$ as the following product of matrices:  
\begin{equation*}
T_-(n,z) = \begin{pmatrix}  I  &   -B(n) \\ 0 &  I  \end{pmatrix} \begin{pmatrix} I &  zI \\ 0  & I \end{pmatrix}  \begin{pmatrix} 0 & -I \\ I  & 0 \end{pmatrix} := T_B T_zT_J\,.
\end{equation*}
Since $J$ is a symplectic matrix, so is $-J$.  Moreover the matrix $T_B\in SL(2d,\R)$.  So $T_J(Z)$ and $T_B(W)$ are in $\mathcal{S}_d $ for any $Z$ and $W$ in $\mathcal{S}_d$.  Claim: $T_z(Z)\in\mathcal{S}_d$ for any $Z\in\mathcal{S}_d$. But this follows easily since $T_z(Z) =Z +zI$ which is clearly in $\mathcal{S}_d$.
\end{proof} 

Next we observe the distance properties of  Weyl $m$ functions. In one dimensional space, hyperbolic metrics are commonly used to analyze the value distribution for solutions of Schr\"odinger equations.

Define the following map on $\mathcal{S}_d$, $d_\infty:\mathcal{S}_d\times\mathcal{S}_d\to\R$ via	
\begin{equation}\label{ms}
d_{\infty}(Z_1, Z_2) = \inf_{Z(t)} \int _0^1 F_{Z(t)}(\dot{Z}(t))\,dt\,, \,\, Z_1, Z_2\in\mathcal{S}_d
\end{equation}
where
\begin{equation}\label{norm} 
F_Z(W) = \|Y^{-1/2}WY^{-1/2}\|\,,
\end{equation} 
and where the infimum is taken over all differentiable paths $Z(t)$ joining $Z_1$ to $Z_2$. Here $Y$ is a positive definite matrix and possesses the square root and $Y^{-1/2} = (Y^{1/2})^{-1} = (Y^{-1})^{1/2}$.  The norm in equation (\ref{norm}) is the operator norm of matrices acting on $\C^d$.

A calculation shows that $d_\infty$ is a metric on $\mathcal{S}_d$ and hence $(\mathcal{S}_d, d_\infty)$ is a metric space.

\begin{lemma}\label{pdm}
For any $d\times d$ positive definite matrix with real entries and any positive $\lambda$:
\begin{equation*}
A^{1/2}(A+\lambda I)^{-1/2} = (A(A+\lambda I)^{-1})^{1/2} = ((I+\lambda A^{-1})^{-1})^{1/2}\,.
\end{equation*}
Moreover, if $I-\lambda A$ is positive, then
\begin{equation*}
\|A^{1/2}(A+\lambda I)^{-1/2}\|^2 < \frac{1}{1+\lambda^2}\,.
\end{equation*}
\end{lemma}
 
\begin{proof}
The first equation follows from the continuous functional calculus with the square root function on $[0,\infty)$ and the fact that $\lambda>0$.  The norm inequality now follows from the first equation and the usual norm estimate on a Neumann series of operators.
\end{proof}
 
For any $T\in SL(2d,\R)$, when considering the map from (\ref{flt}) as a group action, \cite{KAM} showed that this action, in fact, is distance preserving with respect to the metric $d_\infty$. That is,
 
\begin{theorem}\label{distpres}
For any $T\in SL(2d,\R)$ and any $W_1, W_2\in\mathcal{S}_d$ we have 
\begin{equation*}
d_\infty(T(W_1), T(W_2)) = d_\infty(W_1, W_2)\,.
\end{equation*}
\end{theorem}

In general if $T\in SL(2d,\C)$, Theorem \ref{distpres} need not be true.  However, the map given in equation (\ref{flt}) with transfer matrices $T_{\pm}(n)$ is distance decreasing.  We have the following theorem.

\begin{theorem}\label{mth}
Let $z \in \C^+$ and suppose $W_j = T_- (0, z)Z_j$ for $Z_1$ and $Z_2$ in $\mathcal{S}_d$.  Then 
\begin{equation}\label{dp1}
d_\infty(T_-(n,z) W_1, T_-(n,z) W_2)\le\frac{1}{1+ y^2}\,d_\infty(W_1, W_2)\,,
\end{equation}
for all $n\in\N$.
\end{theorem}

\begin{proof}
As in the proof of Lemma \ref{lem2.6} Write $T_-(n,z) =  T_B  T_z T_J$. Since $T_B$ and $T_J$ are in $SL(2d,\R)$, and the fact that $T_z$ has range in $\mathcal{S}_d$, by Theorem \ref{distpres}, we get 
\begin{equation}\label{2.271}
\begin{aligned}
d_\infty(T_-(n,z)(W_1), T_-(n,z)(W_2)) & = d_\infty(T_B T_z T_J(W_1), T_B T_z T_J(W_2)) \\ 
& = d_\infty(T_z T_J(W_1), T_z T_J(W_2))\,.
\end{aligned}
\end{equation} 
Moreover since $T_J\in SL(2d,\R)$ we have
\begin{equation*}
d_\infty(T_J(W_1), T_J(W_2)) = d_\infty(W_1, W_2)\,.
\end{equation*}
It therefore follows that $T_-(n,z)(W)\in\mathcal{S}_d$. 

Let $U_j = T_J(W_j) = JW_j$, for $j=1,2$.  Then $T_z(U_j) = U_j +zI$. Let $U(t)$ be a length minimizing path between $U_1$ and $U_2$. Then $U(t)+zI$ is a path between $U_1 +zI$ and $U_2 +zI$. Moreover
\begin{equation*}
\frac{d}{dt}(U(t) +zI) = \dot{U(t)}\,.
\end{equation*}
Suppose $V(t)= \operatorname{Im} U(t)$ then $\operatorname{Im}\big(U(t) +zI \big) = V(t) + yI.$ We have the following calculation:
\begin{equation*}
\begin{aligned}
&F_{(U(t) +zI)}(\dot{U(t)})= \\
&= \|(V(t) + yI)^{-1/2}\dot{U(t)}(V(t) + yI^{-1/2}\| \\ 
&= \|(V(t) + yI)^{-1/2}V(t)^{1/2}V(t)^{-1/2}\dot{U(t)}V(t)^{-1/2}V(t)^{1/2}(V(t) + yI)^{-1/2}\| \\ 
&\le\|(V(t) + yI)^{-1/2}V(t)^{1/2}\|\cdot\|V(t)^{-1/2}\dot{U(t)}V(t)^{-1/2}\|\cdot\|V(t)^{1/2}(V(t) + yI)^{-1/2}\|\,. 
\end{aligned}
\end{equation*}
Since $V(t)^{1/2}$ and $(V(t) + yI)^{-1/2}$ are symmetric, it follows that
\begin{equation*}
\|(V(t) + yI)^{-1/2}V(t)^{1/2}\| = \|V(t)^{1/2}(V(t) + yI)^{-1/2}\|\,.
\end{equation*}
Therefore, from the symmetry and the above inequality, we have 
\begin{equation}\label{2.27}
F_{(U(t) +zI)}(\dot{U(t)})\le\|(V(t) + yI)^{-1/2}V(t)^{1/2}\|^2\|V(t)^{-1/2}\dot{U(t)}V(t)^{-1/2}\|\,.
\end{equation}

Next, since $U_j = JW_j$, for $j=1,2$, $U(t)$ must be of the form $JW(t)$ where $W(t)$ is a path between $W_1$ and $W_2$.  However, since $W_j = T_-(0, z)Z_j$, $W$ is of the form $W = T_-(0, z)Z $ where $Z(t)= X(t)+iY(t)$ is a path between $Z_1$ and $Z_2$.  By viewing the multiplication as the group action, we have the following
\begin{equation*}
\begin{aligned}
W(t) &=T_-(0, z)Z(t) = \begin{pmatrix} z I -B(0) &  I \\ - I  & 0 \end{pmatrix} Z(t) = zI -B(0)- Z(t)^{-1}\\
& = zI -B(0)- \overline{Z(t)}|Z(t)|^{-1}\,,
\end{aligned}
\end{equation*}
where $|Z| = X^2 +Y^2$.  Since $Y$ is positive and $\operatorname{Im}W(t) = yI + Y(X^2+Y^2)^{-1}$, it follows that $\operatorname{Im}W(t)-yI$ is also positive.  Thus by applying the continuous functional calculus to $\operatorname{Im}W(t)-yI$ with the function $f(x)=1/x$ on $(0,\infty)$, we get 
\begin{equation}\label{imag_ineq}
(\operatorname{Im} W(t))^{-1}-\frac{1}{y}I<0\,.
\end{equation}

Since $U(t) = JW(t)= -W(t)^{-1}$, by applying the continuous functional calculus again we have
\begin{equation*}
U(t) = -W(t)^{-1} = -\overline{W(t)}|W(t)|^{-1} = -\overline{W(t)}\left((\operatorname{Re}W(t))^2 + (\operatorname{Im} W(t))^2\right)^{-1}\,.\end{equation*}
Therefore, 
\begin{equation*}
\operatorname{Im}U(t) =  \operatorname{Im}W(t)\left((\operatorname{Re} W(t))^2 + (\operatorname{Im} W(t))^2\right)^{-1}\,.
\end{equation*}
Since 
\begin{equation*}
(\operatorname{Im}W(t))^2\le(\operatorname{Re} W(t))^2 + (\operatorname{Im} W(t))^2\,,
\end{equation*}
using the functional calculus again with the function $f(x)=1/x$ on $(0,\infty)$ and using inequality (\ref{imag_ineq}) we obtain,
\begin{equation*}
\operatorname{Im}W(t)\left((\operatorname{Re} W(t))^2 + (\operatorname{Im} W(t))^2\right)^{-1}\le\operatorname{Im} W(t)^{-1} < \frac{1}{y}I\,. \end{equation*}
This implies that $I-y\operatorname{Im}U(t)$ is positive and hence $I-y\operatorname{Im}V(t)$ is positive. By Lemma \ref{pdm} we get 
\begin{equation*}
\|(V+yI)^{-1/2}V^{1/2}\| < \frac{1}{1+y^2}\,.
\end{equation*}
By using the above equation (\ref{2.27}) becomes 
\begin{equation*}
\begin{aligned}
F_{(U(t) +zI)}(\dot{U(t)}) &\le\|(V(t) + yI)^{-1/2}V(t)^{1/2}\|^2 \|V(t)^{-1/2}\dot{U(t)} V(t)^{-1/2}\| \\
&\le\frac{1}{1+y^2}F_{(U(t))}(\dot{U(t)})\,.
\end{aligned}
\end{equation*}
Integrating the above inequality yields:
\begin{equation*}
\int_0^1 F_{(U(t)+z)}(\dot{U(t)})\,dt \le\frac{1}{1+y^2}\int_0^1 F_{(U(t))}(\dot{U(t)})\,dt\,.
\end{equation*}
Then taking the infimum over all such paths $U(t)$ we obtain the following:
\begin{equation*}
d_\infty(T_z U_1, T_z U_2) \le\frac{1}{1+ y^2}\,d_\infty(U_1, U_2)\,.
\end{equation*}
Again since $J\in SL(2d,\R)$, we have
\begin{equation*}
d_\infty(U_1, U_2) = d_\infty(T_JW_1, T_JW_2) = d_\infty(W_1, W_2)\,.
\end{equation*}
Using the previous two equations in equation (\ref{2.271}) we obtain:
\begin{equation*}
d_\infty(T_-(n,z) W_1, T_-(n,z) W_2) \le\frac{1}{1+ y^2}\,d_\infty(W_1, W_2)
\end{equation*}
which completes the proof.
\end{proof}
 
Let $P_-(n,z) = T_-(n,z)\cdots T_-(1,z)$.  Then, since $T_-(j,z)$ are symplectic matrices for $j=1,\ldots,n$ and $SL(2d,\C)$ is a group, it follows that $P_-(n,z)\in SL(2d,\C)$.  We have the following corollary.

\begin{cor}
For any $z\in\C^+$, 
\begin{equation*}
d_\infty(M_-(n,z), P_-(n,z) M_+(0,z))\le\frac{1}{(1+ y^2)^n}\,d_\infty(M_-(0,z), M_+(0,z))\,.
\end{equation*}
\end{cor} 

\begin{proof}
For any $z\in\C^+$ and $n\in\N$ by Lemma \ref{mfunctrans} we have 
\begin{equation*}
M_-(n,z) = P_-(n,z) M_-(0,z)\,.
\end{equation*}
Hence the result follows from Theorem \ref{mth}.
\end{proof}
 
We define another $m$ function on the left half line $(-\infty, n)$.  Assume that $F_-(n,z)$ is a square summable matrix-valued solution to equation (\ref{ds}) defined by equation (\ref{ss}) on $(-\infty, n)$.  Define $\widetilde{M}_-(n,z)$ by the following:
\begin{equation*}
\widetilde{M}_-(n,z) = -F_-(n-1,z)F_-(n,z)^{-1}\,.
\end{equation*}
Similar to the proofs in Lemma \ref{lem2.2} and Proposition \ref{prop2.3}, we see that $\widetilde{M}_-(n,z)$ is symmetric and $\operatorname{Im}\widetilde{M}_-(n,z)$ is a positive operator. In addition, $\widetilde{M}_-(n,z)$ and $M_-(n,z)$ are related by the following equation:
\begin{equation*}
\widetilde{M}_-(n,z) = M_-(n,z) -zI + B(n)\,.
\end{equation*}
For fixed $n\in\N$, the mappings $z\mapsto M_{\pm}(n,z)$ are matrix valued Herglotz functions. Set $M_+(z) =  M_+(0,z)$. Since $M_+$ is a  matrix-valued Herglotz function, by Riesz-Herglotz representation, there exists a matrix-valued measure $\mu$ on the bounded Borel subset of $\R$ so that 
\begin{equation*}
M_+(z) = C + Dz + \int_{\R} \left(\frac{1}{\lambda-z} - \frac{\lambda}{1+\lambda^2}\right)\,d\mu_+ \,,
\end{equation*}
for some constants $C$ and $D$. In addition,  $M_+(z)$ has finite normal limits, that is
\begin{equation*}
M_+(\lambda\pm i0) = \lim_{\varepsilon\downarrow 0} M_+(\lambda\pm i\varepsilon)
\end{equation*}
for almost all $\lambda\in\R$, see \cite{SF, FG}. More on matrix-valued Herglotz functions are explained  in the same papers. By the Lebesgue Decomposition Theorem, the matrix-valued Borel measure $\mu+$ can be decomposed as:
\begin{equation*}
\mu_+ = \mu_{+,ac} + \mu_{+,sc} +\mu_{+,pp}\,,
\end{equation*}
where $\mu_{+,ac}$ is the measure that is absolutely continuous with respect to $\mu_+$, $\mu_{+,sc}$ is the singular continuous part of $\mu_+$ and $\mu_{+,pp}$ is the pure point part of $\mu_+$.  Let $\Sigma$, $\Sigma_{ac}$, $\Sigma_{sc}$, and $\Sigma_{pp}$ be the topological support of $\mu_+$, $\mu_{+,ac}$, $\mu_{+,sc}$, and $\mu_{+,pp}$ respectively. 

The following theorem from \cite{FG} showed the connection between the matrix-valued Herglotz functions and the essential support of corresponding spectral measures.  We state it without proof.
\begin{theorem}
\begin{equation*}
\begin{aligned}
\Sigma_{ac} &= \bigcup_{r=1}^d \left\{\lambda\in\R : \lim_{\varepsilon\downarrow 0} M_+(\lambda + i\varepsilon)\textrm{ exists, } \operatorname{rank }\operatorname{Im}(M_+(\lambda + i0)) = r \right\} \\
\Sigma_{sc} &= \bigcup_{r=1}^d \left\{\lambda\in\R : \lim_{\varepsilon\downarrow 0} \operatorname{Im}(\operatorname{tr } M_+(\lambda + i\varepsilon))=\infty,\,\lim_{\varepsilon\downarrow 0}\varepsilon\operatorname{tr }M_+(\lambda + i\varepsilon)=0 \right\} \\	 
\Sigma_{pp} & = \bigcup_{r=1}^d \left\{\lambda\in\R : \operatorname{rank }\lim_{\varepsilon\downarrow 0}\varepsilon M_+(\lambda + i\varepsilon) = r \right\}\,.
\end{aligned}
\end{equation*}
\end{theorem}
 
Describing the absolutely continuous spectrum is one of the main goals in studying vector-valued discrete Schr{\"o}dinger operators.  Remling successfully did this for one dimensional Schr{\"o}dinger operators, \cite{CR1}, and in his renowned paper \cite{CR}, he describes it for Jacobi matrices.  The techniques utilized there involve the concept of reflectionless.  An operator is said to be {\it reflectionless} on a Borel set $B\subseteq\R$ if the associated Herglotz functions satisfy $m_+(x)=-\overline{m_-(x)}$ for almost all $x\in B$ in the Lebesgue sense.  This definition naturally extends to our matrix-valued Herglotz fuctions $M_\pm$.  The above theorem is also utilized in describing the absolutely continuous spectrum as yields precise limit conditions on the Herglotz functions.  Describing the absolutely continuous spectrum is an extremely difficult question and this is one of our main goals in future work.  

{\bf Acknowledgment:}
The first author  was partially supported by the Office of Sponsored Research at Embry-Riddle Aeronautical University, Florida.


\begin{thebibliography}{99}

\bibitem{KA1}K. R. Acharya, Titchmarsh-Weyl theory for vector-valued discrete Schr{\"o}dinger operators, {\it Anal. Math. Phys.}, 9, 1831–1847, 2019.
		
\bibitem{KAM} K. R. Acharya,  M. McBride, Action of complex symplectic matrices on the Siegel upper half space, {\it Lin. Alg. App.}, 563, 47–62, 2019.

		

\bibitem{SF} S. L. Clark and F. Gesztesy, Weyl-Titchmarsh M-function asymptotics for matrix valued Schr\"odinger operators, {\it Proc. Lon. Math. Soc.}, 82, 701-720, 2001.
		
\bibitem{CGH} S. L. Clark and F. Gesztesy, H. Holden, and B. M. Levitan,  Borg-type theorems for matrix-valued Schr\"odinger operators, {\it Jour. Diff. Eq.}, 167, 181-210, 2000. 
				
\bibitem{DD} D. Damanik, A. Pushnitski, B. Simon, The analytic theory of matrix orthogonal polynomials. {\it Surv. Approx. Theo.}, 4, 1-85, 2008.		
		
\bibitem{JFRA} J. Eckhardt, F. Gesztesy, R. Nichols, A. Sakhnovich, and G. Teschl,  Inverse spectral problems for Schr\"odinger-type operators with distributional matrix-valued potentials,
{\it Differential Integral Equations} 28 (2015), 505–522.
		
\bibitem{JG} J. S. Geronimo, Scattering theory and matrix orthogonal polynomials on the real line, {\it Cir. Sys. Sig. Proc.}, 1, 472-495, 1982.
		
\bibitem{FG} F. Gesztesy, E. Rsekanovskii, On matrix-valued Herglotz functions. {\it Math. Machr.}, 218, 61-138, 2000.
		
\bibitem{RK} R. Kozhan, Equivalence classes of block Jacobi matrices. {\it Proc. Amer. Math. Soc.}, 139, 799-805, 2011.
				
\bibitem{CR} C. Remling, The absolutely continuous spectrum of Jacobi Matrices, {\it Ann. Math.}, 174, 125-171, 2011.
		
\bibitem{CR1} C. Remling, The absolutely continuous spectrum of one-dimensional Schr\"odinger operators, {\it Math. Phys. Anal. Geom.}, 10, 359-373, 2007.
	
\bibitem{GT} G. Teschl, Jacobi Operators and Completely Integrable Nonlinear Latices, Mathematical Monographs and Surveys,  Vol.72, American Mathematical Society, Providence, 2000.

\end{thebibliography}
\end{document}